\documentclass[conference]{IEEEtran}

\usepackage{latexsym,amsmath,amssymb}
\def\F {{\mathbb{F}}}

\def\Fq {{\mathbb{F}_q}}
\def\C{\mathbb{C}}

\def\R{\mathbb{R}}

\def\ba{{\bf a}}
\def\bb{{\bf b}}

\def\bu{{\bf u}}
\def\bv{{\bf v}}

\def\bo{{\bf 0}}

\def\bx{{\bf x}}
\def\by{{\bf y}}

\def\a{{\alpha}}

\def\beq{\begin{equation}}
\def\eeq{\end{equation}}

\newtheorem{thm}{Theorem}[section]
\newtheorem{defn}[thm]{Definition}
\newtheorem{prop}[thm]{Proposition}
\newtheorem{lem}[thm]{Lemma}
\newtheorem{cor}[thm]{Corollary}
\numberwithin{equation}{section} \newtheorem{rem}[thm]{Remark}

\begin{document}

\title{Quantum Gilbert-Varshamov Bound Through Symplectic Self-Orthogonal Codes }

\author{Lingfei~Jin and Chaoping~Xing
\thanks{L.  Jin and C.  Xing are with Division of
Mathematical Sciences, School of Physical and Mathematical Sciences,
Nanyang Technological University, Singapore 637371, Republic of
Singapore (email:ljin1@e.ntu.edu.sg, xingcp@ntu.edu.sg).}
\thanks{The work  was partially supported by the Singapore  National Research
Foundation under Research Grant NRF-CRP2-2007-03.}
\thanks{C. P. Xing is the corresponding author.}}

\maketitle

\begin{abstract} It is well known that quantum codes can be constructed through classical  symplectic self-orthogonal codes.
In this paper, we give a kind of Gilbert-Varshamov bound for  symplectic self-orthogonal codes first and then obtain the Gilbert-Varshamov bound for quantum codes. The idea of obtaining the Gilbert-Varshamov bound for  symplectic self-orthogonal codes follows from counting arguments.
\end{abstract}
\begin{keywords} Symplectic self-orthogonal, Quantum Gilbert-Varshamov bound, Symplectic distance.
\end{keywords}

\section{Introduction}
In the past few years, the theory of quantum error-correcting codes have been developed rapidly. Various constructions are given through classical coding. However, it is still a great challenge to construct good quantum codes. Shor and Steane gave the first construction of quantum codes through classical self-orthogonal codes. Subsequently, constructions of quantum codes through classical codes with certain self-orthogonality have been extensively studied and investigated. For instance, quantum codes can be obtained from Euclidean, Hermitian self-orthogonal or Symplectic self-orthogonal codes  (see \cite{Aly Kla Sar, Ash Kni, Bie Ede, Fen Ma, Gra Bet, Ket Kla}).

As in the classical coding theory, the quantum Gilbert-Varshomov (GV, for short) bound is also a benchmark for good quantum codes. The first quantum GV  bound was obtained in \cite{Ash Ba Kni Lit} for the binary case. One year later, Ashikhmin and Knill \cite{Ash Kni} generalized the binary quantum GV bound to the  $q$-ary case. In 2004,  Feng and Ma \cite{Fen Ma} derived a finite version of Gilbert-Varanmov bound for classical Hermitian self-orthogonal codes  and then applied to quantum codes to obtain a finite version of quantum GV bound.

In 1998, Macwilliams and Sloane \cite{Mac Slo1} showed through counting arguments that binary self-dual codes can achieve the GV bound for classical case. In this paper, we use the counting arguments as well to show that classical symplectic self-orthogonal codes can achieve the GV bound  and then apply to quantum codes to obtain the asymptotic quantum GV bound.

Our paper is organized as follows. We first recall some basic notations and results of symplectic slef-orthogonal codes in Section II. In section III, we present a kind of GV bound for symplectic self-orthogonal codes through counting arguments. In section IV, we apply our result obtained in Section III to quantum codes and derive the asymptotic quantum GV bound.

\section{Preliminaries}
 With the development of classical error-correcting codes, people have extensively studied the Euclidean inner product and investigated the Euclidean self-orthogonal codes. However, due to applications to quantum codes in recent years, other inner products  such as Hermitian and symplectic inner products have attracted researchers in this area and many intersecting results have been obtained already. In this section, we introduce some basic results and notations about symplectic inner products and symplectic slef-orthogonal codes.

Let $\Fq$ be a finite field with $q$ elements, where $q$ is a prime power.  For four vectors $\ba=(a_1,\dots,a_n), \bb=(b_1,\dots,b_n), \ba^{'}=(a_1^{'},\dots,a_n^{'}), \bb^{'}=(b_1^{'},\dots,b_n^{'})\in \F_q^n$, the symplectic inner product $\langle,\rangle_S$ is defined by
\[\langle(\ba|\bb),(\ba^{'}|\bb^{'})\rangle_S=\langle\ba, \bb^{'}\rangle_E-\langle\bb, \ba^{'}\rangle_E,\]
where $\langle,\rangle_E$ is defined as the ordinary dot inner product (or Euclidean inner product).
For an $\Fq$-linear code $C$ in $\F_q^{2n}$, define the symplectic dual of $C$ by
\[C^{\perp_S}=\{(\bx|\by):\; \langle(\bx|\by),(\ba|\bb)\rangle_S=0\ \mbox{ for all}\; (\ba|\bb)\in C\}.\]
It is easy to show that $\dim_{\Fq}(C)+\dim_{\Fq}(C^{\perp_S})=2n$.
A code $C$ is said symplectic self-orthogonal if $C\subseteq C^{\perp_S}$, and self-dual if $C= C^{\perp_S}$.

For an vector $(\ba|\bb)=(a_1,\dots,a_n|b_1,\dots,b_n)\in\F_q^{2n}$, the symplectic weight is defined by :
\[{\rm wt}_S(\ba|\bb)=|\{i:\;(a_i,b_i)\neq(0,0), i=1\dots n\}|\]
For two vectors $(\ba|\bb), (\ba^{'}|\bb^{'})\in\Fq^{2n}$, the symplectic distance is defined by :
\[d_S((\ba|\bb), (\ba^{'}|\bb^{'}))={\rm wt}_S(\ba-\ba^{'}|\bb-\bb^{'}).\]

The symplectic minimum distance of a linear code $C\in\F_q^{2n}$ is defined by
\[d_S(C)=:\min\{{\rm wt}_S(\ba|\bb):\; (\ba|\bb)\in C-\{(\bo|\bo)\}\}.\]

Then it is straightforward to verify that a $[2n,k]$-linear code $C$ also satisfies the symplectic Singleton bound:
\[k+2d_S(C)\le 2n+2.\]

\section{Gilbert-Varshamov Bound for symplectic self-orthogonal codes }
In the previous section, we saw the symplectic Singleton bound already. Similarly, we can derive  the symplectic GV bound. However, for application to quantum codes, we are interested in a GV type bound for symplectic self-orthogonal codes. The main goal of this section is to derive such a bound through counting argument.

First, we have the following simple but useful observation.

\begin{lem}\label{3.1} Every vector in $\F_q^{2n}$ is orthogonal to itself with symplectic inner product.
\end{lem}

\begin{lem}\label{3.2}  The number of symplectic self-orthogonal codes of length $2n$ and dimension $k$ in the vector space $\F_q^{2n}$ is given by
\begin{equation}\label{eq3.1}\frac{(q-1)^{k-1}(q^{2n-2k+2}-1)(q^{2n-2k+4}-1)\dots(q^{2n}-1)}{(q^k-1)(q^{k-1}-1)\dots(q-1)}.\end{equation}
\end{lem}

\begin{proof} For convenience, denote by $A_k$ the number of symplectic self-orthogonal codes of length $2n$ and dimension $k$ and denote  by $C_{2n,k}$ a symplectic self-orthogonal code of length $2n$ and dimension $k$ over $\Fq$. Then for $k\le n$,  $C_{2n,k-1}$ can be extended to $C_{2n,k}$ by adding a vector $\bu\in C_{2n,k-1}^{\perp_S}\setminus C_{2n,k-1}$. Thus, we can obtain $|C_{2n,k-1}^{\perp_S}/ C_{2n,k-1}|-1=q^{2n-2(k-1)}-1$ distinct symplectic self-orthogonal codes of length $2n$ and dimension $k$ from $C_{2n,k-1}$ through this way.

On the other hand, by the above argument, we know that every symplectic self-orthogonal code of length $2n$ and dimension $k-1$ is contained in a symplectic self-orthogonal code of length $2n$ and dimension $k$. Since every subspace of dimension $k-1$ of $C_{2n,k}$ is symplectic self-orthogonal and there are $(q^k-1)/(q-1)$ subspaces of dimension $k-1$ in $C_{2n,k}$, we get a recursive formula
\[A_{k}=\frac{(q-1)(q^{2n-2(k-1)}-1)}{q^k-1}A_{k-1}\]
for $k=2,\dots n$.
The desired result follows from the above recursive formula and the fact that $A_1=\frac{q^{2n}-1}{q-1}$.
\end{proof}

\begin{lem}\label{3.3} Given a nonzero vector $\bu\in\F_q^{2n}$, the number of symplectic self-orthogonal codes containing $\bu$ of length $2n$ and dimension $k$  is given by
\begin{equation}\label{eq3.2} \frac{(q-1)^{k-1}(q^{2n-2k+2}-1)(q^{2n-2k+4}-1)\dots(q^{2n-2}-1)}{(q^{k-1}-1)(q^{k-2}-1)\dots(q-1)}.\end{equation}
\end{lem}
\begin{proof} Similarly, we denote by $C_{2n,k}(\bu)$ a symplectic self-orthogonal code containing $\bu$ of length $2n$ and dimension $k$ over $\Fq$ and denote by $B_k(\bu)$  the number of such $C_{2n,k}(\bu)$.
Using the similar arguments as in lemma \ref{3.2}, we can establish the following recursive formula for $B_k(\bu),$
\[B_{k}(\bu)=\frac{(q-1)(q^{2n-2(k-1)}-1)}{q^{k-1}-1}B_{k-1}(\bu).\]
The desired result follows from the above recursive formula and the fact that $B_1(\bu)=1$.
\end{proof}

\begin{thm}[GV bound]\label{3.4} For $1\le k\le n$ and  $d\geq 1$, there exists a $[2n,k,d]$ symplectic self-orthogonal code over $\F_q$ if
\[\sum_{i=1}^{d-1}{n\choose i}(q^2-1)^i<\frac{q^{2n}-1}{q^k-1}.\]
\end{thm}

\begin{proof} As $B_k(\bu)$ is independent of $\bu$, we denote by $B_k$ the number given in (\ref{eq3.2}).

First,  the number of nonzero vectors $\bu$ of symplectic weight less than $d$ is given by
\[V(2n,d)=:\sum_{i=1}^i{n\choose i}(q^2-1)^i.\]

 On the other hand, there are at most
$\sum_{{\rm wt}_S(\bu)<d}B_k(\bu)=V(2n,d)B_k$ symplectic self-orthogonal codes with symplectic distance less than $d$.
Therefore, there is at least one symplectic self-orthogonal code with symplectic distance at least $d$ if $A_k>V(2n,d)B_k$.
The desired result follows from this inequality and Lemmas \ref{3.2} and \ref{3.3}.
\end{proof}

Theorem \ref{3.4} shows existence of symplectic self-orthogonal codes with good symplectic distance. However, to construct good quantum codes through symplectic self-orthogonal codes, we have to control symplectic dual distance for a given symplectic self-orthogonal code.

\begin{lem}\label{3.5} For a given nonzero vector $\bu\in\F_q^{2n}$, the number $E_k$ of symplectic self-orthogonal codes $C$ of length $2n$ and dimension $k$ such that $C^{\perp_S}$ contain $\bu$  satisfies the following recursive formula{\small
\begin{equation}\label{eq3.3} E_k=\frac{(q-1)(q^{2n-2(k-1)-1}-1)}{q^k-1}E_{k-1}+\frac{(q-1)^2q^{2n-2(k-1)-1}}{q^k-1}B_{k-1}\end{equation}}
for any $k\geq2$, where $B_k$ is the quantity defined in (\ref{eq3.2}).
\end{lem}
\begin{proof} Let $W$ be the symplectic dual space of $\langle\bu\rangle$. Then $C^{\perp_S}$ contains $\bu$ if and only if $C$ is a subspace of $W$. Thus, $E_k$ stands for the  number of symplectic self-orthogonal codes $C$ of length $2n$ and dimension $k$ such that $C\subseteq W$. We  denote by $D_k$ the  number of symplectic self-orthogonal codes $C$ of length $2n$ and dimension $k$ such that $\bu \in C\subseteq W$.  We also denote by $F_k$ the  number of symplectic self-orthogonal codes $C$ of length $2n$ and dimension $k$ such that $\bu \not\in C\subseteq W$. Then it is easy to see that $D_k=B_k$ and $D_k+F_k=E_k$.

Without confusion, we denote  by $C_{2n,k}$ a symplectic self-orthogonal code of length $2n$ and dimension $k$ over $\Fq$ such that $C\subseteq W$. Then for $k\le n$,  $C_{2n,k-1}$ can be extended to $C_{2n,k}$ by adding a vector $\bv\in W\cap(C_{2n,k-1}^{\perp_S}\setminus C_{2n,k-1})$. Thus, we can obtain  $|C_{2n,k-1}^{\perp_S}/ C_{2n,k-1}|-1=q^{2n-2(k-1)}-1$ (or $q^{2n-2(k-1)-1}-1$, respectively) distinct symplectic self-orthogonal codes of length $2n$ and dimension $k$ from $C_{2n,k-1}$ through this way if $\bu\in C_{2n,k-1}$ (or if $\bu\not\in C_{2n,k-1}$, respectively).

On the other hand, by the above argument, we know that every symplectic self-orthogonal codes of length $2n$ and dimension $k-1$ is contained in a symplectic self-orthogonal codes of length $2n$ and dimension $k$. Since every subspace of dimension $k-1$ of $C_{2n,k}$ is symplectic self-orthogonal and there are $(q^k-1)/(q-1)$ subspaces of dimension $k-1$ in $C_{2n,k}$, we get a recursive formula
\[\frac{q^k-1}{q-1}E_k=(q^{2n-2(k-1)}-1)D_{k-1}+(q^{2n-2(k-1)-1}-1)F_{k-1}.\]
The desired  reclusive formula follows.
\end{proof}
\begin{cor}\label{3.6} Let $E_k$ stands for the same number defined in Lemma \ref{3.5}. Then one has
\begin{equation}\label{eq3.4} E_k\le k\frac{(q-1)^{k-1}(q^{2n-2k+1}-1)\cdots(q^{2n-1}-1)}{(q^k-1)\cdots(q-1)}\end{equation}
for all $k\ge 1$.
\end{cor}
\begin{proof} We know that $E_1=\frac{q^{2n-1}-1}{q-1}$. So it is true for $k=1$.

Now assume that the result is also true for $k-1$. Then by Lemma \ref{3.5}, we have
\begin{eqnarray*}
E_k&=&\frac{(q-1)(q^{2n-2(k-1)-1}-1)}{q^k-1}E_{k-1}+\\
&&\frac{(q-1)^2q^{2n-2(k-1)-1}}{q^k-1}B_{k-1}\\
&\le&(k-1)\frac{(q-1)^{k-2}(q^{2n-2k+1}-1)\cdots(q^{2n-1}-1)}{(q^k-1)\cdots(q-1)}+\\ &&\frac{(q-1)^2q^{2n-2(k-1)-1}}{q^k-1}B_{k-1}\\
&=&(k-1)\frac{(q-1)^{k-2}(q^{2n-2k+1}-1)\cdots(q^{2n-1}-1)}{(q^k-1)\cdots(q-1)}+\\ &&\frac{(q-1)^kq^{2n-2k+1}(q^{2n-2k+4}-1)\cdots(q^{2n-2}-1)}{(q^k-1)\cdots(q-1)}\\
&\le&k\frac{(q-1)^{k-1}(q^{2n-2k+1}-1)\cdots(q^{2n-1}-1)}{(q^k-1)\cdots(q-1)}.
\end{eqnarray*}
This completes the proof.
\end{proof}
By using the same arguments in the proof of Theorem \ref{3.4}. we obtain the following result.
\begin{cor}\label{3.7} For $1\le k\le n$ and  $d\geq 1$, there exists a $[2n,k]$ symplectic self-orthogonal code $C$ over $\F_q$ with $d_S(C^{\perp_S})\ge d$ if
\[\sum_{i=1}^{d-1}{n\choose i}(q^2-1)^i<\frac{\prod_{i=0}^{k-1}(q^{2n-2i}-1)}{k\prod_{i=0}^{k-1}(q^{2n-2i-1}-1)}.\]
\end{cor}

\section{ Application to quantum Gilbert-Varshamov Bound }
For a quantum code $Q$, we denote by  $n(Q), K(Q), d(Q)$ {length}, {dimension} and  {minimum distance} of $Q$, respectively. A fundamental domain $U_q^Q\subseteq[0,1]\times[0,1]$  is defined as follows.

\begin{defn}\label{4.1} $U_q^Q$ is a set which is consisted of ordered pairs $(\delta, R)\in [0,1]\times[0,1]$ such that there exists a family of $q$-ary quantum codes $\{Q_i\}_{i=1}^\infty$ satisfying{\small
 \[n(Q_i)\rightarrow\infty,
 R=\lim_{i\rightarrow\infty}\log_qK(Q_i)/n(Q_i), \delta=\lim_{i\rightarrow\infty}d(Q_i)/n(Q_i).\]}
\end{defn}
 As in the classical coding theory, determining the domain  $U_q^Q$ is one of the central  asymptotic problem for quantum coding theory. One can imagine that it is very hard to completely determine $U_q^Q$.  Nevertheless, some bounds on $U_q^Q$ have been given by many researchers \cite{Ash Kni, Chen Lin Xin, Fen Lin Xin}. A useful description of $U_q^Q$ through a function $\a_q^Q$ was given by Feng-Ling-Xing \cite{Fen Lin Xin}:
\begin{quote}there
exists a function $\alpha_{q}^Q(\delta), \delta\in[0,1]$, such that
$U_q^Q$ is the union of the domain
$$\{(\delta, R)\in\R^{2}:\ 0\le R<\alpha_q^Q(\delta),\; 0\le\delta\le 1\}$$
with some points on the boundary $\alpha_q^Q(\delta)$.
\end{quote}
Thus,
determining $U_q^Q$ is almost equivalent to determining the function
$\alpha_{q}^Q(\delta)$.

To apply our results in the previous section, let us establish a connection between symplectic self-orthogonal codes and quantum codes.

\begin{lem}\label{4.2} \cite{Ash Kni}
If $C$ is a $q$-ary symplectic self-orthogonal $[2n,k]$ code, then there exists a $q$-ary $[[n, n-k,
d]]$ quantum code with $d=d_S(C^{\perp_S})$.
\end{lem}
Thus, by combining Corollary \ref{3.7} with Lemma \ref{4.2}, we immediately get the following result.
\begin{cor}\label{4.3} For $1\le k\le n$ and  $d\geq 1$, there exists a $q$-ary $[[n,n-k,d]]$ quantum code $C$ if
\[\sum_{i=1}^{d-1}{n\choose i}(q^2-1)^i<\frac{\prod_{i=0}^{k-1}(q^{2n-2i}-1)}{k\prod_{i=0}^{k-1}(q^{2n-2i-1}-1)}.\]
\end{cor}

Finally we are ready to derive the quantum GV bound.
\begin{thm}\label{4.4}
\[
\alpha_{q}^Q(\delta)\ge 1-\delta\log_q(q+1)-H_q(\delta).
\]
where $H_q(x)$ is the $q$-ary entropy function $x\log_q(q-1)-x\log_qx-(1-x)\log_q(1-x)$.
\end{thm}
\begin{proof}
Fix $\delta\in (0,1)$. For every $n\ge 1/\delta$, denote by $d$ the quantity $\lfloor\delta n\rfloor$. Then $d/n$ tends to $\delta$ when $n$ tends to $\infty$. Put
\[k=\left\lfloor \log_q\left(\sum_{i=1}^{d-1}{n\choose i}(q^2-1)^i\right)+\log_qk\right\rfloor +1.\]
Then we have
\[\sum_{i=1}^{d-1}{n\choose i}(q^2-1)^i<\frac{q^k}{k}<\frac{\prod_{i=0}^{k-1}(q^{2n-2i}-1)}{k\prod_{i=0}^{k-1}(q^{2n-2i-1}-1)}.\]
Hence, we have a $q$-ary quantum $[[n,n-k,d]]$-code by Corollary \ref{4.3}. Therefore, one has
\[
\frac{n-k}{n}=1-\frac kn\rightarrow 1-\delta\log_q(q+1)-H_q(\delta).
\]
The proof is completed.
\end{proof}


\begin{thebibliography}{99}
\bibitem{Aly Kla Sar} S. A.~Aly, A.~Klappenecker and
P. K.~Sarvepalli, ``On quantum and classical BCH codes,"
\emph{IEEE Trans. Inf. Theory,} vol. 53, no. 3, pp. 1183--1188,
Mar. 2007.

\bibitem{Ash Kni} A.~Ashikhmin and E.~Knill, ``Nonbinary quantum stablizer
codes," \emph{IEEE Trans. Inf. Theory,} vol. 47, no. 7, pp.
3065--3072, Nov. 2001.

\bibitem{Ash Ba Kni Lit} A.~Ashikhmin, A. M. Barg, E.~Knill and S. N. Litsyn, ``Quantum error detection II: Bounds,"
\emph{IEEE Trans. Inf. Theory,} vol. 46, no. 3, pp.
789--800, May. 2000.

\bibitem{Ash Lit} A.~Ashikhmin and S.~Litsyn, ``Upper bounds on the size of quantum codes,"
\emph{IEEE Trans. Inf. Theory,} vol. 45, no. 4, pp. 1206--1215,
May 1999.

\bibitem{Ash Lit Tsf}
A.~Ashikhmin, S.~Litsyn and M. A. Tsfasman, ''Aymptotically good quantum codes,'' Phys. Rev. A, vol. 63, no. 3, p. 032311, Mar. 2001

\bibitem{Bie Ede} J.~Bierbrauer and Y.~Edel, ``Quantum twisted
codes," \emph{J. Comb. Designs,}  vol. 8, pp. 174--188, 2000.

\bibitem{Cal Rai} A. R.~Calderbank, E. M.~Rains, P. W.~Shor and
N. J. A.~Sloane, ``Quantum error correction via codes over GF(4),"
\emph{IEEE Trans. Inf. Theory,} vol. 44, no. 4, pp. 1369--1387,
July 1998.

\bibitem{Che Lin Xin} H.~Chen, S.~Ling and C.~Xing, ``Quantum codes from concatenated
algebraic-geometric codes," \emph{IEEE Trans. Inf. Theory,} vol. 51, no. 8, pp. 2915--2920, Aug.
2005.

\bibitem{Chen Lin Xin}
H. Chen, S. Ling, and C. P. Xing, ''Aymptoticially good quantum codes exceeding the Ashikhmin-Litsyn-Tsfasman bound,'' \emph{IEEE. Trans. Inform. Theory,} vol. 47, no. 5, pp. 2055-2058, Jul. 2001.

\bibitem{Fen Lin Xin}
K. Q. Feng, S. Ling and C. P. Xing, "Aymptotic Bounds on Quantum Codes From Algebraic Geometry Codes,"
\emph{IEEE. Trans. Inform. Theory,} vol. 52, no. 3, pp. 986-991, Mar. 2006.

\bibitem{Fen Ma}
K. Q. Feng and Z. Ma, ``A finite Gilbert-Varshmov Bound for pure stabilizer quantum codes,"
\emph{IEEE. Trans. Inform. Theory,} vol. 50, no. 12, pp. 3323-3325, Dec. 2004.

\bibitem{Gra Bet} M.~Grassl and T.~Beth, ``Quantum BCH codes,"
International Symposium on Theoretical Electrical Engineering, Magdeburg, 1999, pp. 207-212.

\bibitem{Ket Kla} A.~Ketkar, A.~Klappenecker, S.~Kumar and P.~Sarvepalli, ``Nonbinary stablizer codes over finite fields," \emph{IEEE.
Trans. Inform. Theory,} vol. 52, no. 11,  pp. 4892--4914, Nov.
2006.

\bibitem{Kni Laf} E.~Knill and R.~Laflamme, ``A theory of quantum error-correcting codes," \emph{Phys. Rev. A,}
vol. 55, no. 2, pp. 900--911, 1997.

\bibitem{Li Xin Wan1} Z.~Li, L. J.~Xing and X. M.~Wang, ``A family of asymptotically good quantum codes
based on code concatenation," \emph{IEEE. Trans. Inform. Theory,}
vol. 55, no. 8,  pp.3821--3824, Aug. 2009.

\bibitem{LX04} S. Ling and C. P. Xing, {\it Coding Theory -- A
First Course,} Cambridge University Press, 2004.

\bibitem{Mac Slo1}
F. J. Macwilliams and N. J. Sloane, "Good Self-Dual codes Exists," Discrete Mathematics, vol. 3,
 no. 1,2,3, pp. 153-162, Sep, 1972.


\bibitem{Rain} E. M. Rains, ``Nonbinary quantum codes," \emph{IEEE. Trans. Inform. Theory,} vol. 45, no. 6,
pp.1827--1832, Sept. 1999.

\bibitem{Si} J. H. Silverman, {\it The Arithmetic of Elliptic Curves},
Springer, New York, 1986.

\bibitem{Stea3} A. M. Steane, ``Enlargement of Calderbank-Shor-Steane quantum
codes," \emph{IEEE. Trans. Inform. Theory,} vol. 45, no. 7, pp.
2492--2495, Nov. 1999.

\bibitem{M. A. S. G.}
M. A. Tsfasman and S. G. Vl$\breve{a}$dut, ''Algebraic-Geometric Codes. Amsterdam,'' The Netherlands:Kluwer, 1991.




\end{thebibliography}
\end{document}